\documentclass[11pt]{article}

\usepackage[a4paper, margin=1in]{geometry}
\usepackage{parskip}
\usepackage{graphicx}
\usepackage{subcaption}
\usepackage{algorithm}
\usepackage{algorithmic}
\usepackage{comment}
\usepackage{amsmath,amssymb,amsfonts, amsthm}
\usepackage{url}
\usepackage{float}
\usepackage{makecell}
\usepackage{colortbl}
\usepackage{xcolor}
\usepackage{hyperref}
\usepackage{tikz}
\usepackage{caption}
\theoremstyle{plain}
\newtheorem{theorem}{Theorem}[section]

\newtheorem{lemma}[theorem]{Lemma}

\theoremstyle{definition}

\theoremstyle{remark}

\newcommand{\defn}{\textbf}

\newcommand{\pred}{\hat{\sigma}}
\newcommand{\ind}{i}
\newcommand{\G}{\tilde{G}}

\title{\textbf{Incremental Strongly Connected Components \\ with Predictions}}

\author{
Ronald Deng \\ Williams College \\ \texttt{rd8@williams.edu}
\and 
Samuel McCauley \\ Williams College \\ \texttt{srm2@williams.edu}
\and
Aidin Niaparast\thanks{This material is based upon work supported in part by the Air Force Office of Scientific Research under award number FA9550-23-1-0031.} \\ Carnegie Mellon University \\ \texttt{aniapara@andrew.cmu.edu}
\and 
Helia Niaparast \\ Carnegie Mellon University \\ \texttt{hniapara@andrew.cmu.edu} 
\and 
Bennett Ptak \\ Williams College \\ \texttt{bp14@williams.edu}
\and 
Shirel Quintanilla \\ Williams College \\ \texttt{zsq1@williams.edu}
\and
Shikha Singh \\ Williams College \\ \texttt{shikha.singh@williams.edu}
\and 
Nathan Vosburg \\ Williams College \\ \texttt{njv1@williams.edu}
}

\date{}

\begin{document}

\maketitle

\begin{abstract}
Algorithms with predictions is a growing area that aims to leverage machine-learned predictions to design faster beyond-worst-case algorithms.  
In this paper, we use this framework to design a learned data structure for the incremental strongly connected components (SCC) problem.  In this problem, the $n$ vertices of a graph are known a priori and the $m$ directed edges arrive over time. The goal is to efficiently maintain the strongly connected components of the graph after each insert.  Our algorithm receives a possibly erroneous prediction of the edge sequence and uses it to precompute partial solutions to support fast inserts. We show that our algorithm 
achieves nearly optimal bounds with good predictions and its performance smoothly degrades with the prediction error.  We also implement our data structure and perform experiments on real datasets. Our empirical results show that the theory is predictive of practical runtime improvements.
\end{abstract}

\section{Introduction}\label{sec:intro}
Incremental computation on graphs is a fundamental primitive in many database and optimization problems. In an incremental graph problem, the vertices of the graph $G$ are known ahead of time and the edges are revealed over time.  The goal is to design an efficient data structure for answering queries at time $t$ about some property (such as connectivity or shortest path) of the current graph $G_t$.  
Designing fast \emph{worst-case} algorithms for incremental graph problems is an area of
extensive research~\cite{chen2024almost,bender2015new,bernstein2021incremental,fan2017incremental,haeupler2012incremental}. 

Algorithms designed to be fast in the worst case provide strong guarantees against adversarial inputs.  However, these algorithms fail to exploit the correlations in the structure of typical instances.  As most computations are performed repeatedly on similar datasets, heuristics optimized for domain specific inputs tend to be faster in practice.   

A growing body of work, known as algorithms with predictions~\cite{mitzenmacher2022algorithms}, seeks to bridge the gap between simple heuristics that lack strong theoretical guarantees and more complex worst-case algorithms that are provably efficient but often impractical. The main idea is to use machine-learned predictions to capture correlations in the input instances and leverage them to guide algorithmic decisions.

Predictions are particularly effective for online or dynamic problems, where the input is not known ahead of time.  Predictions trained on historical data can then serve as a proxy for future inputs---turning a hard \emph{online} problem into an easier-to-solve \emph{offline} one.  Using such predictions, an algorithm can then precompute partial solutions ahead of time and consequently perform faster updates and queries at runtime.  The algorithms designed in this framework are referred to as \emph{learning-augmented} or \emph{learned} algorithms.  

The running time of a learned algorithm is measured in terms of the quality of the prediction---the bound is stated in terms of both the size of the input and the prediction error.  Ideally, a learned algorithm for an online problem (1) is nearly as fast as the best-offline solution when the prediction is perfect, (2) has bounded worst-case performance cost when the prediction is arbitrarily wrong, and (3) has performance that degrades gracefully with error.  These properties in order are referred to as \emph{consistency, robustness}, and \emph{smoothness} respectively in the literature.  Thus, learned algorithms leverage good predictions to be fast and have provable safeguards against bad predictions.

\subparagraph*{Incremental SCC.}  In this paper, we study the \emph{incremental strongly connected components (SCC)} problem in the algorithms with predictions model.  In this problem, the $n$ vertices $V$ of a directed graph $G$ are known a priori and the $m$ directed edges $E$ are inserted one by one. Let $e_1, \dots, e_m$ denote the sequence of edges and let $G_t= (V, \{e_1, \ldots, e_t\})$ denote the graph at time $t$. 
The goal is to maintain the \textbf{strongly connected components} of $G_t$ at all times $t$.
A strongly connected component $S \subseteq V$ is a maximal set of vertices of $G$ such that every pair of vertices $u, v$ are \emph{mutually reachable}, that is, there is a directed path from $v$ to $u$  and vice versa.  

Maintaining SCCs is a fundamental building block in many applications, e.g. social networks~\cite{dhingra2016finding}, graph processing~\cite{zhang2017efficient}, program analysis~\cite{hardekopf2009pointer} and fault-tolerant networks~\cite{baswana2019efficient}.

DFS-based approaches such as Tarjan's algorithm~\cite{tarjan1972depth} can compute the SCCs of a static offline graph $G$ in $O(n+m)$ time.  
As a step towards solving the incremental SCC problem, consider the \emph{offline incremental problem}. In this variant,  the edge insert sequence $\sigma = e_1, e_2, \ldots, e_m$ is known ahead of time and the goal is to design a data structure to answer queries about the SCCs of each intermediate graph $G_t$.  Naively running Tarjan's repeatedly takes $O(m^2)$  time.\footnote{We assume $m \geq n$ to simplify the runtimes.}  A folklore divide-and-conquer approach~\cite{sccblog} improves upon this naive approach and solves the offline incremental SCC problem in $O(m \log m)$ total time.

\subparagraph*{Offline to Learned Incremental SCC.}
To design a learned algorithm for the \emph{online incremental problem}, we assume that a prediction $\hat{\sigma}$ of the edge sequence $\sigma$ is given to the algorithm. We define the \defn{prediction error} $\eta$ as the maximum error between when an edge was predicted to arrive in $\hat{\sigma}$ and its actual arrival time in $\sigma$. To ensure this metric is well-defined, we assume that $\hat{\sigma}$ is a permutation of $\sigma$ (we discuss this assumption further in Section~\ref{sec:model}). If the prediction is perfect  ($\eta = 0$), the problem reduces to the offline incremental variant.  If the prediction is arbitrarily erroneous ($\eta = m$), the algorithm effectively gets no information about the input and the problem reduces to the worst-case model. The main challenge is to design an algorithm that smoothly interpolates between the two regimes as a function of the prediction error $\eta$.  

\subsection{Our Contributions}
One of our main contributions is designing the first learned data structure for the incremental SCC problem which is consistent, robust and smooth with respect to the prediction error.  

\begin{theorem}\label{thm:runtime}
We design a learned data structure that, given prediction $\hat{\sigma}$ with error $\eta$, maintains the strongly connected components of the graph in total time $O(m \log m \cdot \max\{\eta, 1\})$ for $m$ inserts. Queries about the SCC of a given vertex take $O(1)$ time.
\end{theorem}

To put this result in context, notice that if the edge arrivals are reasonably predictable (e.g. $\eta = O(\text{polylog}(m)))$, our algorithm for the online incremental problem is nearly as fast as the offline variant.

The algorithmic idea is to use a divide-and-conquer approach on the predicted edge sequence $\hat \sigma$ to recursively precompute solutions to partial subproblems using Tarjan's algorithm.  As edges arrive that are different from $\hat \sigma$, our algorithm updates the prediction and lazily recomputes the affected subproblems. At a high level, we show that we can bound the number of times each subproblem is recomputed by the prediction error $\eta$ and thus lose an $\eta$ factor over the offline incremental algorithm.   

\subparagraph*{Robustness to Error.}
The running time of our algorithm grows with $\eta$---in fact, for sufficiently high $\eta$, our algorithm may be slower than known worst-case techniques.  
However, the algorithm can be made robust to arbitrarily erroneous predictions with respect to any worst-case algorithm by switching to the worst-case
algorithm if the learned algorithm’s runtime grows larger than the worst-case guarantee.

Similarly, our algorithm's running time is given in terms of the maximum error---as stated, even a single outlier can significantly impact performance guarantees.  One potential strategy to handle a small number of outliers is to reset the algorithm (restarting from time zero with the outlier predicted correctly) each time an outlier arrives. Each reset is as expensive as an entire algorithm run, but between resets, performance is proportional to the maximum error among non-outliers. While this strategy is expensive on a per-outlier basis, it can be used to handle a small number of outliers without degrading the performance too much. In our experiments,  our algorithm (without this intervention)  seems to be somewhat resistant to outliers---its performance on our datasets is better than the maximum error would indicate.

\subparagraph*{Comparison with Worst-Case Approaches.}  Most prior algorithms for maintaining SCCs in incremental graphs do so by maintaining \emph{topological ordering} of the \emph{condensed graph} (a directed acyclic graph in which all vertices within an SCC are contracted into a new single node). 
A topological ordering is a labeling $L: V \rightarrow \mathbb{Z}$ of the vertices $V$ such that $L(v) < L(u)$ if there is a directed path from $v$ to $u$. A topological ordering exists if and only if the directed graph is acyclic. 
An algorithm that maintains the topological ordering can detect new cycles whenever a new edge points from a node with a higher label to a lower label.  Such a cycle leads to the merging of previous (condensed) SCCs. 

The best-known worst case total running times of algorithms using the topological-ordering-based combinatorial approach are $\tilde{O}(n^2)$~\cite{bender2015new} for dense graphs and $\tilde{O}(m^{4/3})$~\cite{bernstein2021incremental} for sparse graphs. Note that the $\tilde{O}$-notation hides logarithmic factors. 
Thus, even with reasonably large error as long as $\eta < \min\{n^2/m, m^{1/3}\}$, our algorithm improves upon these approaches.

In a recent theoretical advancement, Chen et al.~\cite{chen2024almost} use the interior-point method to maintain SCCs in incremental graphs (without a topological ordering) in total time $m \cdot e^{O(\log^{167/168}m \log \log m}$).   
While the improvement is theoretically appealing, the algorithm and techniques employed are not practical and the running time itself is galactic.

Faster algorithms are known for maintaining SCCs in the \emph{decremental} setting (where all edges are available at the beginning and deleted one by one).  This can be done in total expected time ${O}(m \log^4 m)$~\cite{bernstein2019decremental}.  Predictions of the edge sequence essentially help turn an incremental SCC problem (which seems to be harder) into a decremental one. 

In contrast to all prior work, our algorithm maintains SCCs without maintaining a topological ordering or using complex techniques.  Our algorithm exploits predictions to precompute partial solutions to recursive subproblems.  Our recursive technique that uses predictions to lift the offline problem to an online one follows the approach used by McCauley et al.~\cite{McCauleyMoNi25} for maintaining approximate shortest paths in incremental graphs using predictions.  

One way our paper differs from the previous works on dynamic graph problems with predictions~\cite{BrandFNP24, liu2023predicted, McCauleyMoNi25, HenzingerSSY24} is that we co-designed the algorithm with its implementation.  Concretely, our analysis explicitly tracks logarithmic factors in running time, and avoids constant-factor overheads when possible (e.g.\ rebuilding unnecessary subproblems as in~\cite{McCauleyMoNi25}).  

\subparagraph*{Experimental Results.}
We compare the runtime of our algorithm to the state-of-the-art heuristic $\text{IncSCC}^+$~\cite{fan2017incremental} on real datasets\footnote{We ended up optimizing $\text{IncSCC}^+$ beyond what was intended in~\cite{fan2017incremental}. We include runtime comparison against both variants.}. $\text{IncSCC}^+$ maintains a {topological ordering} of the {condensed graph} and optimizes over running Tarjan's from scratch repeatedly. Despite the optimizations, their worst-case total cost is $\Theta(m^2\log m)$ for $m$ inserts. 

Our experiments show that (1) our algorithm is scalable and can handle large datasets, (2) it outperforms the known recursive algorithm for the offline problem under perfect predictions, (3) it significantly outperforms $\text{IncSCC}^+$ on reasonably accurate predictions,  and finally (4) its runtime degrades gradually with error.  These results complement our theoretical analysis: they give evidence that the algorithm is reasonably simple, and its running time does not involve large constants.

\subparagraph*{Summary.} In summary, our theoretical and experimental results demonstrate that predictions can help bridge the gap between theoretically-good algorithms and practical heuristics.  Our results provide a proof-of-concept for designing algorithms with strong provable guarantees that can match or even outperform state-of-the-art heuristics on predictable inputs. 

\subsection{Additional Related Work}
\subparagraph*{Algorithms with Predictions For Runtime.}
Kraska et al.~\cite{KraskaBCDP18} jump-started the area by \emph{empirically} demonstrating how machine-learned predictions can speed up data structures.  Dinitz et al.~\cite{DinitzILMV21} present a \emph{theoretical} framework for analyzing predictions for speeding up \emph{offline} algorithms.  This has since been used for problems such as maximum flow~\cite{DaviesMVW,POLAK2024106487}, shortest path~\cite{LattanziSV23}, minimum cut~\cite{niaparast2025faster}, stable matching~\cite{mccauley2026stable}, and convex optimization~\cite{SakaueO22}. 

Predictions have been used to design provably faster data structures for problems such as online list labeling~\cite{McCauleyMNS23}, incremental topological ordering ~\cite{mccauley2024incremental}, binary search~\cite{dinitz2024binary}, sorting~\cite{bai2023sorting}, dictionaries~\cite{zeynali2024robust,lin2022learning}, and priority queries~\cite{benomar2024learningaugmented}.  Similar to this paper, predictions have been used to ``lift'' offline solutions to dynamic graph problems to improve the runtime of incremental graph problems, such as incremental shortest paths~\cite{BrandFNP24,HenzingerSSY24,McCauleyMoNi25} and divide-and-conquer algorithms~\cite{liu2023predicted}. 
We note in particular that our divide-and-conquer framework is similar to that of McCauley et al.~\cite{McCauleyMoNi25}, however, we keep only a single path rather than maintaining the entire tree.  This improves runtime performance by several log factors and also makes the data structure more practical.
Liu et al.~\cite{liu2023predicted} also use a divide-and-conquer framework, but their analysis requires balanced subproblem sizes.  In contrast, we can handle unbalanced subproblems with the trade-off that our error bounds are on max error (rather than average).
Overall, these results demonstrate significant potential of leveraging machine-learned predictions to speed up algorithms and data structures.  
\section{Preliminaries}\label{sec:prelim}
This section presents the formal setup and a warm-up algorithm for the offline incremental SCC problem.

\subsection{Model and Notation} 
\label{sec:model}

Let $V$ be the set of vertices of the directed graph $G$.  
Let $\sigma = e_1, \ldots, e_m$ denote the sequence of online edge inserts. 
We assume that one edge arrives per time step, so we often say that $e_t$ arrives at ``time $t$.''
We use ``graph at time $t$'', denoted by $G_t$, to refer to the graph with vertex set $V$ and the first $t$ edges $e_1, \ldots, e_t$. We define $G_0$ as the graph on vertex set $V$ with no edges.
We assume that $G_m$ has no isolated vertices; if it is not then all isolated vertices in $G_m$ can be removed before running the algorithm, increasing the cost of the algorithm by $|V|$.
Before any edges arrive, the learned algorithm receives a prediction $\hat{\sigma}$ of $\sigma$.  We define $m = |\sigma| = |\hat{\sigma}|$.  
Let $\ind(e)$ be the index of $e$ in $\sigma$ and $\widehat{\ind}(e)$ be the index of $e$ in $\hat{\sigma}$.
Let $\eta_e = |\ind(e) - \widehat{\ind}(e)|$ for each edge $e$ in $\sigma$.
The error $ \eta = \max_{e \in \sigma} \eta_e $ is the maximum error of any edge; for simplicity of exposition, we assume $\eta \geq 1$ (the analysis is asymptotically the same if $\eta = 0$ or $\eta = 1$).
As we go through the edges, we update the predictions, and we denote the predicted arrival sequence at time $t$ by $\hat{\sigma}_t$.
This is the same prediction model as past work on incremental graph algorithms~\cite{HenzingerSSY24,BrandFNP24,McCauleyMoNi25}. 

We assume that the set of edges in $\pred$ and $\sigma$ are the same---that is, $\pred$ is a permutation of $\sigma$.  
It is possible to generalize this assumption, albeit with some extra cost.  If an edge $e$ arrives that is not in $\pred$, the algorithm can simply start over from the beginning, at a cost of $O(m\log m)$---this does not increase the overall cost if the number of such edges is at most $\eta$.  
Edges in $\pred$ that never arrive are more expensive, since effectively $\eta_e = m$. One strategy to handle such edges is to use a large threshold $\tau$: at each time step $t$, the algorithm can remove any edge that was predicted to arrive before $t - \tau$; if the edge later arrives it can be treated as an edge that is not in $\pred$. 

Our goal is to maintain a data structure that, after edge $e_t$ arrives, can answer queries to determine if $u$ and $v$ are in the same SCC of $G_t$, for any $u,v \in V$.  

\subsection{Warm Up: Offline Incremental SCC}\label{sec:warmup}

Our algorithm is based closely on the offline recursive algorithm for the problem.
For the offline incremental problem, the entire sequence of edge insertions $\sigma$ is known ahead of time.  The goal is to quickly process $\sigma$ to allow us to answer SCC queries for $G_t$ for any $t$.    
We describe a folklore recursive method to perform this preprocessing efficiently.  This algorithm has been described in programming contest literature, see e.g.,~\cite{sccblog}, and a similar structure has been used for designing a worst-case algorithm for decremental SCC~\cite{roditty2013decremental}.

The algorithm recursively divides the interval $[0,m]$ into halves.  Each interval generated in this process corresponds to a \defn{subproblem}. 
In particular, the resulting intervals form a hierarchical set of dyadic subintervals (for ease of exposition, we assume that 
$m$ is a power of two, although this assumption is not necessary for the algorithm or its correctness). 
Each subproblem has two endpoints $\ell$ and $r$.
We refer to a subproblem either by its endpoints, $[\ell, r]$, or by its midpoint, $x := (r + \ell)/2$. The subproblem  $[\ell, r]$ has a \defn{left child} $[\ell, x]$ and a \defn{right child} $[x, r]$, and is referred to as their \defn{parent} subproblem. See Figure~\ref{fig:tree} for a schematic view of the recursion tree.

For each interval $[\ell, r]$, the algorithm maintains a graph $\tilde{G}_x = (V_x, E_x)$.  
$\tilde{G}_x$ is defined recursively below, but to help motivate the definition let us briefly summarize its contents.  Consider contracting all SCCs of $G_{\ell}$ into single vertices, retaining only edges between two distinct SCCs.  These vertices and edges are enough to determine what SCCs are created during $[\ell, r]$.  Now, consider trimming the graph further: we 
remove all edges that do not contribute to forming a new SCC during $[\ell, r]$, and 
remove all isolated vertices from the graph; trimming the graph in this way still allows us to determine the correct SCCs during $[\ell, r]$.  This is the idea behind $\tilde{G}_x$.

At the root of the recursion tree, $\ell = 0$, $r = m$ and $\G_{x} = G_{m/2}$.  The algorithm runs Tarjan's algorithm~\cite{tarjan1972depth} on $G_x$ to compute its SCCs. It uses the SCCs to compute the graphs for the left and right subproblems $\tilde{G}_{x}^{\text{left}}$ and $\G_{x}^{\text{right}}$  respectively.   We define $\tilde{G}_x^{\text{left}}$ using its vertex set $V_x^{\text{left}}$ and edge set $E_x^{\text{left}}$; $\tilde{G}_x^{\text{right}}$ is defined likewise.

\subparagraph*{Right Subproblem.} Consider a pair of vertices $u$ and $v$ in the same strongly connected component $S$ of $\G_{x}$.  Since edges are only ever inserted, $u$ and $v$ will remain in the same SCC at all time steps after $x$.  Thus, to construct the right subproblem, the algorithm contracts all nodes within the same SCC in $\G_x$ into a single node and deletes all edges $(u,v)$ such that $u$ and $v$ are within the same SCC.  In particular, for the right subproblem corresponding to the interval $[x, r]$, it sets $V_{x}^{\text{right}}$ to contain a contracted node for each strongly connected component in $\G_{x}$. The edge set $E_{x}^{\text{right}}$ includes all edges $(u,v) \in E_x$ such that $(u,v)$ is an \emph{inter SCC} edge with endpoints $u$ and $v$ in different SCCs in $\G_{x}$. 

\subparagraph*{Left Subproblem.}  Now consider an edge $(u,v)$ in $\G_{x}$ with end points in different SCCs of $\G_{x}$.  Then, the vertices $u$ and $v$ will be in different SCCs at all time steps before $x$. The algorithm can safely delete such edges when recursing on the left subproblems.  That is, $V_{x}^{\text{left}}$ contains all vertices in $V_{x}$ with at least one incident edge, and $E_{x}^{\text{left}}$ only contains the \emph{intra-SCC} edges in $E_{x}$ with both end points in the same SCC.  

\subparagraph*{Recursion Tree.}
The algorithm recurses on the left and right subproblems until the interval $[\ell, r]$ has length two. This way, each time $t=1,\ldots,m-1$ is the midpoint of exactly one subproblem in the recursion tree.
Notice that the nodes of this recursion tree in an in-order traversal correspond to the sequence of graphs $G_0, \ldots, G_m$; see Figure~\ref{fig:tree}.

\subparagraph*{Queries.}
The folklore solution~\cite{sccblog} builds a graph to implicitly track SCC merges. Here we describe one simple way to handle the queries recursively instead.  

To see if $u$ and $v$ are in the same SCC at time $t$, we begin at the root and find the descendant of the root containing $t$. If it is a right child, $u$ and $v$ may each be merged into a new vertex.  If so, we replace $u$ (resp. $v$) with the merged vertex, and recurse.  When we reach a subproblem with midpoint $t$, we check if $u$ and $v$ are in the same SCC in $\G_t$, and return the result. This takes $O(\log m)$ time. 

In Section~\ref{sec:nodelabel}, we present a data structure for the online problem that explicitly maintains node labels over time and answers queries in $O(1)$ time.

\subparagraph*{Analysis.} Each time the algorithm processes a subproblem, it recurses on two subproblems obtained by splitting the time interval into two halves, implying a recursion depth of $O(\log m)$. Running Tarjan's algorithm on a subproblem with $s$ edges takes $O(s)$ time.  At each recursive call, an edge is either an inter- or intra-SCC edge, and thus is only passed down to one of the left or right subproblems. Thus, each edge contributes to one subproblem per level and the total runtime is $O(m)$ per level of the recursion tree. The total insertion cost is thus $O(m\log m)$.

\begin{figure}[t]
\centering
\begin{tikzpicture}[
    x=0.62cm, y=0.8cm,
    levelbox/.style={draw=black, rounded corners=2pt, minimum height=0.6cm},
    slot/.style={draw=black, minimum width=0.62cm, minimum height=0.55cm, inner sep=0pt},
    >=stealth
]

    % level 0
    \draw[levelbox] (0,3) rectangle (16,3.6);
    \node (L0) at (8,3.3) {\scriptsize $[0,16]$};

    % level 1
    \draw[levelbox] (0,2) rectangle (8,2.6);
    \draw[levelbox] (8,2) rectangle (16,2.6);
    \node (L10) at (4,2.3) {\scriptsize $[0,8]$};
    \node (L11) at (12,2.3) {\scriptsize $[8,16]$};

    % level 2
    \draw[levelbox] (0,1) rectangle (4,1.6);
    \draw[levelbox] (4,1) rectangle (8,1.6);
    \draw[levelbox] (8,1) rectangle (12,1.6);
    \draw[levelbox] (12,1) rectangle (16,1.6);
    \node (L20) at (2,1.3) {\scriptsize $[0,4]$};
    \node (L21) at (6,1.3) {\scriptsize $[4,8]$};
    \node (L22) at (10,1.3) {\scriptsize $[8,12]$};
    \node (L23) at (14,1.3) {\scriptsize $[12,16]$};

    % level 3
    \draw[levelbox] (0,0) rectangle (2,0.6);
    \draw[levelbox] (2,0) rectangle (4,0.6);
    \draw[levelbox] (4,0) rectangle (6,0.6);
    \draw[levelbox] (6,0) rectangle (8,0.6);
    \draw[levelbox] (8,0) rectangle (10,0.6);
    \draw[levelbox] (10,0) rectangle (12,0.6);
    \draw[levelbox] (12,0) rectangle (14,0.6);
    \draw[levelbox] (14,0) rectangle (16,0.6);
    \node (L30) at (1,0.3) {\scriptsize $[0,2]$};
    \node (L31) at (3,0.3) {\scriptsize $[2,4]$};
    \node (L32) at (5,0.3) {\scriptsize $[4,6]$};
    \node (L33) at (7,0.3) {\scriptsize $[6,8]$};
    \node (L34) at (9,0.3) {\scriptsize $[8,10]$};
    \node (L35) at (11,0.3) {\scriptsize $[10,12]$};
    \node (L36) at (13,0.3) {\scriptsize $[12,14]$};
    \node (L37) at (15,0.3) {\scriptsize $[14,16]$};

    % edges level 0 -> 1
    \draw (L0.south) -- (L10.north);
    \draw (L0.south) -- (L11.north);

    % edges level 1 -> 2
    \draw (L10.south) -- (L20.north);
    \draw (L10.south) -- (L21.north);
    \draw (L11.south) -- (L22.north);
    \draw (L11.south) -- (L23.north);

    % edges level 2 -> 3
    \draw (L20.south) -- (L30.north);
    \draw (L20.south) -- (L31.north);
    \draw (L21.south) -- (L32.north);
    \draw (L21.south) -- (L33.north);
    \draw (L22.south) -- (L34.north);
    \draw (L22.south) -- (L35.north);
    \draw (L23.south) -- (L36.north);
    \draw (L23.south) -- (L37.north);
\end{tikzpicture}
\caption{Recursive tree of subproblems for $m=16$.}
\label{fig:tree}
\end{figure}
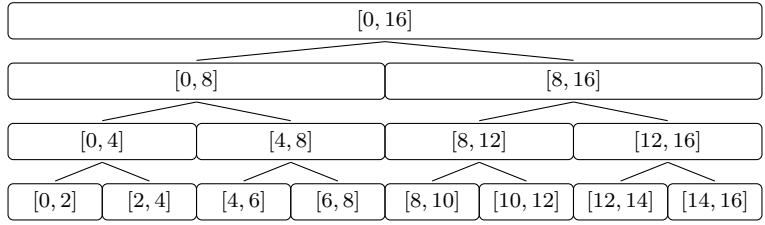
\section{Incremental SCC with Predictions}\label{sec:main_algorithm}\label{sec:algo}
In this section, we describe how we adapt the offline recursive algorithm for the incremental SCC problem with predictions. That is, given a prediction $\hat{\sigma}$ of $\sigma$ at the beginning, we describe how to create a data structure that at each time step $t$, efficiently: (1)  inserts $e_t$ (the $t$-th edge in $\sigma$), and (2) allows answers to SCC queries on $G_t$.

First, notice that if the prediction sequence $\hat{\sigma}$ is \emph{perfect} ($\hat{\sigma} = \sigma$), then the online problem reduces to the offline incremental problem.  However, in general, the prediction $\hat{\sigma}$ could differ from $\sigma$ significantly, and naively redoing the offline recursive approach each time is too expensive. 
Instead, our algorithm \emph{only maintains a single path of the recursion tree} from Section~\ref{sec:warmup}; when each edge arrives, the algorithm updates the path, and updates the SCC of each vertex.  In particular, initially, the algorithm only computes the root-to-left path from subproblem $[0,m]$ to subproblem $[0,0]$ in the offline recursion tree built using the prediction $\hat{\sigma}$.  As edges arrive, it updates the predicted sequence $\hat{\sigma}$ and recomputes the current root-to-leaf path for time step $t$.  In our analysis, we show that we only incur recomputations that we can directly charge to the error in the prediction, that is, the algorithm is only $O(\eta)$-factor worse than the offline incremental algorithm. 

\subsection{Learned Incremental SCC Algorithm}\label{sec:main_incscc}

First, we describe how our algorithm, \textbf{Learned IncSCC}, maintains and uses the prediction.

Our algorithm continuously updates $\hat{\sigma}$ as edges arrive. Let $\hat{\sigma}_t$ denote the updated prediction after $t$ edges have arrived and $\hat{\sigma}_0 = \hat{\sigma}$. Thus $\hat{\sigma}_t$ agrees with $\sigma$ on the first $t$ edges. 
Let $t$ and $\hat{t}$ denote the arrival time step of an edge $e_t$ in $\sigma$ and $\hat{\sigma}_{t-1}$ respectively.
If $\hat{\sigma}_{t-1}$ was correct, then $\hat{t} = t$.  Otherwise, we update $\hat{\sigma}_{t-1}$ to obtain $\hat{\sigma}_t$: we delete $e$ from time $\hat{t}$, and insert it to time $t$.  All edges in $\hat{\sigma}_{t-1}$ between $t$ and $\hat{t}$ are shifted one time step later (note that $t'>t$); this gives $\hat{\sigma}_{t}$.  

Now, we discuss how subproblems are maintained.
At a high level, instead of building the entire recursive tree of the offline algorithm from Section~\ref{sec:warmup} with respect to $\hat{\sigma}_0$, Learned IncSCC is more conservative.  It lazily maintains the only thing needed to answer queries about the SCCs at the current timestep $t$:  the path from the root to $t$ in the recursive tree that corresponds to the updated prediction $\hat{\sigma_t}$.

Before any edges arrive, the algorithm uses the prediction $\hat{\sigma}_0$ and the offline recursive approach to \textbf{build} all subproblems along the path from the root $[0,m]$ to the leaf subproblem $[0,0]$.  We describe the implementation of this build in more detail momentarily.

Now consider step $t$ when edge $e_t$ is about to arrive.  Let $\hat{t}$ be
the time $e_t$ was predicted to arrive in $\hat{\sigma}_{t-1}$.  Before $e_t$ arrives, the algorithm maintains the path from the root to time step $t-1$ of the recursion tree corresponding to $\hat{\sigma}_{t-1}$.  
After $e_t$ arrives, the algorithm updates $\hat{\sigma_t}$ and the root-to-leaf path it maintains as follows.  The algorithm begins with the last subproblem in the current path, and walks backwards through the path to find a subproblem with interval $[\ell, r]$ such that $\hat{t}\in [\ell, r]$. Thus, $[\ell, r]$ is the lowest common ancestor (LCA) of the nodes $t-1$ and $\hat{t}$ in the recursion tree with respect to $\hat{\sigma}_{t-1}$. The algorithm calls the build procedure on this subproblem (and thus recursively on its children) to compute the path from the root to $t$ with respect to $\hat{\sigma}_{t}$.

\subparagraph*{What We Store.}
Our algorithm maintains a path from the root to the current time $t$, implemented as an array of at most $\lceil \log_2 m\rceil$ pointers to subproblems.  

Our algorithm maintains the following for each subproblem $[\ell, r]$ with midpoint $x$: (1) a graph $\G_x = (V_x, E_x)$ whose nodes are contracted nodes of $V$, 
(2) the SCCs of $\G_x$,
(3) a mapping $M_x$ 
from each vertex in $\G_x$ to its \defn{representative vertex} in $V$, and (4) a set of \defn{left edges} and \defn{right edges} that correspond to the edges of the left and right subproblems $E_{x}^{\text{left}}$ and $E_x^{\text{right}}$ as defined in Section~\ref{sec:warmup}.  Each subproblem stores the edges of its children so that the edges do not need to be recomputed when the children are rebuilt.

Finally, our algorithm maintains the {Node Label Data Structure} which stores at all times $t$, a \textbf{label} $L(v)$ for each vertex $v \in V$, such that $L(u) = L(v)$ at time $t$ if and only if $u$ and $v$ are in the same SCC in $G_t$.  
We use this data structure as a blackbox below and defer its implementation to Section~\ref{sec:nodelabel}.

\subparagraph*{Building a Subproblem.}
Below we describe how the algorithm performs a \defn{build} operation at time $t$ on an interval $[\ell, r]$ with midpoint $x$ with respect to the predicted sequence $\hat{\sigma}_t$.  

The algorithm begins by calculating the graph $\G_x$ and the mapping $M_x$.  We split into two cases depending on if $[\ell, r]$ is the root; in the case it is not, we split into further cases depending on if it is a left or right child.

 If $[\ell, r]$ is the root subproblem, then $\G_x = (V_x, E_x)$ where $V_x = V$ and $E_x$ contains all edges in $\hat{\sigma}_{t}$.   The mapping $M_x$ has $M_x(v_i) = v_i$ for all vertices $v_i$.

If $[\ell, r]$ is not the root, let $[\ell^p, r^p]$ denote its current parent node.  Let $\G_x^p = (V_x^p, E_x^p)$ and $M_x^p$ denote the graph and mapping stored by the parent subproblem. There are two cases depending on whether the interval $[\ell, r]$ is the left or right child of $[\ell^p, r^p]$.

If $[\ell, r]$ is the left child of $[\ell^p, r^p]$,  then $V_x$ is the set of vertices in $V_x^p$ with at least one incident edge, i.e., incoming or outgoing, and the edge set $E_x$ is the left edges stored in the parent subproblem $\G_{x}^p$.  The mapping is inherited from the parent, that is, $M_x = M_{x}^p$.   

If $[\ell, r]$ is the right child of $[\ell^p, r^p]$, the algorithm contracts the vertices in the same SCC of $\G_x^p$ to a single node, and these contracted nodes form $V_x$. The edge set $E_x$ is the right edges stored in the parent subproblem $\G_{x}^p$. For each node $c \in V_x$, which corresponds to an SCC in $\G_x^p$, we compute $M_x(c)$ as follows. Let $v$ be an arbitrary node in the SCC corresponding to $c$. We set $M_x(c):=M_x^p(v)$.

Now that $\G_x$ and $M_x$ are computed, the algorithm uses them to finish processing the subproblem.
Similar to the offline case, the algorithm computes the strongly connected components of $\G_x$ using Tarjan's algorithm, considering only the edges in $E_x$ that appear at timesteps at most $x$ according to $\hat{\sigma}_t$.

First, consider the case when the midpoint $x\neq t$.
The algorithm computes the left edges of $\G_x$ as those whose endpoints lie in the same SCC in $\G_x$, and the right edges as those whose endpoints lie in different SCCs.  
These left and right edges are stored in the parent subproblem and inherited by child subproblems (rather than being recomputed when a child is rebuilt).  After computing the left and right edges, the algorithm recursively builds the child containing $t$.

The base case occurs when the midpoint $x = t$. In the base case, the algorithm performs a final merge of the labels of nodes in the same SCC of $\G_x$ in the Node Label Data Structure; we define this merge below.

\subsection{Node Label Data Structure}\label{sec:nodelabel}
The \emph{Node Label Data Structure}  stores, at all times $t$, the SCC of each vertex in $G_t$. It consists of: (1) an array $L$ of size $n$, storing a \emph{label} for each vertex such that $L[i] = L[j]$ if and only if vertices $v_i$ and $v_j$ are in the same SCC in $G_t$; and (2) for each label $\ell$, a doubly linked list of all vertices with label $\ell$ (in other words, a list of all vertices $v_i$ such that $L[i] = \ell$).  

At initialization, $L[i] = i$ for all $i \in [n]$.  

On a query to see if vertices $v_i$ and $v_j$ are in the same SCC, we simply check if $L[i] = L[j]$.  This leads to $O(1)$ query time.  

Given a collection of vertices $V_c$, the \textbf{merge} operation merges the SCCs of all vertices in $V_c$ into a single SCC.  Let $V_c$ consist of $k$ vertices $v_{c_1}, v_{c_2}, \ldots, v_{c_k}$. For $i = 1, \ldots, k-1$, we find labels $L[c_i]$ and $L[c_{i+1}]$, and determine which has a smaller list.  
Let $\ell_1$ be the label of the smaller list (breaking ties arbitrarily) and $\ell_2$ be the other.  We iterate through the list of all vertices with label $\ell_1$; for each vertex $v_j$, we set $L[j] = \ell_2$.  Then, we add all vertices with label $\ell_1$ to the list of vertices with label $\ell_2$.
Since the lists are implemented as linked lists, each merge operation takes time proportional to the number of vertices with label $\ell_1$. 
\section{Running Time Analysis}\label{sec:analysis}
In this section, we analyze the running time of our learned algorithm.  We defer proofs related to correctness to the full version of the paper. 

 To start, we bound what edges in the predicted edge sequence $\hat{\sigma}$ can cause a given subproblem to be rebuilt.

\begin{lemma}
\label{lem:rebuild_edges}
     If a subproblem $[\ell, r]$ with midpoint $x$ using prediction $\pred$ is rebuilt at time $t$, then the edge $e$ that arrives at time $t$ was predicted to arrive in $\hat{\sigma}$ at a time either in $[x, x + \eta]$ or in $[r, r + \eta]$.
 \end{lemma}

   \begin{proof} 
   The subproblem $[\ell, r]$ is rebuilt only when it lies on the path from $\mathrm{LCA}(t, \hat t)$ to $t$ for some edge. There are two ways this can happen:

    \begin{enumerate}
        \item $[\ell, r]$ is the $\mathrm{LCA}(t, \hat t)$ for some pair $(t, \hat t)$.
        This occurs only if $t\leq x$ and $\hat{t} \geq x$. Since $|t - \hat t| \le \eta$, this requires that $\hat t \in [x, x + \eta]$.
        
        \item $[\ell, r]$ lies strictly below the LCA. This only happens if $t$ lies inside $[\ell, r]$, but $\hat t$ lies outside the subtree rooted at $[\ell, r]$; therefore, $\hat{t} \geq r$. 
        Since $[\ell, r]$ is on the path to $t$, we must have $t\leq r$; since $\hat{t} - t \leq \eta$, we have that $\hat{t}\in [r, r + \eta]$.\qedhere
    \end{enumerate}
\end{proof}

Before the next lemma, we motivate the structure of our analysis.  Lemma~\ref{lem:rebuild_edges} states that each subproblem is rebuilt only $O(\eta)$ times; this means that if we can bound the size of all subproblems (in other words, the total number of edges in all subproblems) we are done.  
Indeed, the classic analysis of the offline problem with $\eta = 0$ states that each edge is in at most one subproblem at each level.  This would mean that the total size of all subproblems is $O(m\log m)$; combining with Lemma~\ref{lem:rebuild_edges} would give the desired $O(\eta m\log m)$ bound.

Unfortunately, this analysis is incorrect.  The reason is that as $\pred$ changes, what edge belongs to each subproblem changes as well.  
For example, according to $\pred_0$, the endpoints of an edge $e=(u,v)$ might not be in the same SCC in $G_{m/2}$, which puts $e$ in the right child of the root. But at some later time $t<m/2$, an edge originally predicted to arrive at $t'>m/2$ might arrive that puts $u$ and $v$ in the same SCC in $G_{m/2}$. This would place $e$ in the left child of the root.
Lemmas~\ref{lem:combining_time_in_subproblem} through~\ref{lem:combining_time_changing} bound how often this can happen---i.e., an edge may be a part of $O(\eta \log m)$ builds in total.

For any edge $(u,v)$ and prediction $\pred$, let the \defn{combining time} $C_{u,v}(\pred)$ be the first time in $\pred$ when $u$ and $v$ are in the same SCC.  
If $u$ and $v$ are never in the same SCC in $\pred$, then $C_{u,v}(\pred)$ is undefined.
Lemma~\ref{lem:combining_time_in_subproblem} shows that at any time $t$, an edge $(u,v)$ is only in a subproblem that contains its combining time.  This lemma is useful both in analyzing the running time and proving correctness.

\begin{lemma}\label{lem:combining_time_in_subproblem}
Let $[\ell, r]$ be a subproblem  in the offline recursive tree built on $\pred_t$ that contains edge $e = (u,v)$. If $C_{u,v}(\pred_t)$ is defined, then $C_{u,v}(\pred_t) \in [\ell, r]$; if $C_{u,v}(\pred_t)$ is undefined, then $r = m$.
\end{lemma}

 \begin{proof}
We argue by contradiction.  

First, consider the case where $C_{u,v}(\pred_t)$ is undefined; assume by contradiction that $e$ is in a subproblem $[\ell, r]$ with $r\neq m$.
Let $[\ell, r]$ be the subproblem closest to the root containing $e$ where $r \neq m$.  Then $[\ell, r]$ was a left child---but by definition, this means that $u$ and $v$ were in the same component in the graph of the  parent of $[\ell, r]$; a contradiction.

If $C_{u,v}(\pred_t)$ is defined, 
let $[\ell, r]$ be the subproblem in the current path closest to the root such that $e$ is an edge in $[\ell, r]$, but $C_{u,v}(\pred_t)\notin [\ell, r]$.

If $C_{u,v}(\pred_t) < \ell$, let $[\ell_a, r_a]$ be the ancestor of $[\ell, r]$ with midpoint $\ell$. $[\ell_a, r_a]$ must exist since if $C_{u,v}(\pred_t)$ exists, $\ell > 0$.
Furthermore, $[\ell, r]$ must be a descendant of the right child of $[\ell_a, r_a]$. 
But since $C_{u,v}(\pred_t) < \ell$, $u$ and $v$ are in the same SCC in the midpoint of $[\ell_a, r_a]$, so $e$ cannot be in the right edges of $[\ell_a, r_a]$, contradicting that it is an edge in $[\ell, r]$.

If $C_{u,v}(\pred_t) > r$, let $[\ell_a, r_a]$ be the ancestor of $[\ell, r]$ with midpoint $r$.  $[\ell_a, r_a]$ must exist since if $C_{u,v}(\pred_t)$ exists, $r < m$. Furthermore, $[\ell, r]$ is a descendant of the left child of $[\ell_a, r_a]$.  
But since $C_{u,v}(\pred_t)>r$, $u$ and $v$ are in different SCCs in the midpoint of $[\ell_a, r_a]$, so $e$ cannot be in the left edges of $[\ell_a, r_a]$, contradicting that it is an edge in $[\ell,r]$. 
\end{proof}

Now, we show that the error bound $\eta$ implies that the combining time of any pair of vertices cannot differ significantly as $\pred$ is updated.
First, we need a lemma showing that the way we update predictions does not increase the error of any edge beyond $\eta$.

\begin{lemma}
\label{lem:updating_predictions}
    For any edge $e$ and any $\pred_t$, let $i_t(e)$ be the position of $e$ in $\pred_t$, and $i(e)$ be the position of $e$ in $\sigma$.  Then if $\pred_0$ has maximum error $\eta$, we have $|i(e) - i_t(e)| \leq \eta$---thus, $\pred_t$ also has maximum error $\eta$. This implies that for each $t_1$ and $t_2$ we have $|i_{t_1}(e)-i_{t_2}(e)|\leq 2\eta$.
\end{lemma}

\begin{proof} 
We prove this by induction on $t$; it is trivially true for $t=0$.

Note that when we move from $\pred_{t-1}$ to $\pred_{t}$, we move an edge from $\hat{t}\geq t$ to its final location $t$; all edges in positions from $t$ to $\hat{t}-1$ move forward one, and all others remain unchanged.  By definition, $\hat{t} - t \leq \eta$. 
The edge $e$ being moved to $t$ has $i_{t}(e) = i(e)$.  

Any edge $e$  being moved forward one (i.e. an edge between $t$ and $\hat{t}-1$ in $\pred_{t-1}$) that has $i(e) \geq \hat{t}$ is being moved closer to $i(e)$; by the inductive hypothesis the lemma holds.  
Now suppose an edge $e$ is being moved forward one, and $i(e) \leq \hat{t}$.
For all the edges that move forward we have $i_t(e) \leq \hat{t}$. Furthermore, since $e$ has not arrived, $i(e) > t$.  Combining with $\hat{t} - t \leq \eta$ implies $|i_{t}(e) - i(e)| \leq \eta$.
\end{proof}

Finally, Lemma~\ref{lem:combining_time_changing} states that the combining time does not change significantly as we update the predictions.

\begin{lemma}\label{lem:combining_time_changing}
For any pair of vertices $u, v$, and any $\pred_t$, $C_{u,v}(\pred_t) \in [C_{u,v}(\pred_0) - 2\eta, C_{u,v}(\pred_0) + 2\eta]$.
\end{lemma}

\begin{proof} 
   Consider the graph $G'$ constructed using all vertices in $V$, and the first $C_{u,v}(\pred_0)$ edges in  $\hat{\sigma}_0$; $u$ and $v$ are in the same strongly connected component in $G'$.  
   Let $G''$ be the graph constructed using vertices in $V$, and the first $C_{u,v}(\pred_0) + 2\eta$ edges in $\pred_t$.  By Lemma~\ref{lem:updating_predictions}, all edges in $G'$ must have already arrived by time $C_{u,v}(\pred_0) + 2\eta$ in $\pred_t$, so the edges of $G''$ are a superset of the edges of $G'$; that is, $u$ and $v$ must be in the same SCC in $G''$.  Therefore, $C_{u,v}(\pred_t) \leq C_{u,v}(\pred_0) + 2\eta$.
   
   Let $G'$ be the graph  at time $C_{u,v}(\pred_0)- 1$ for $\pred_0$; $u$ and $v$ are not in the same strongly connected component in $G'$.  Let $G''$ be the graph at time $C_{u,v}(\pred_0) - 2\eta - 1$ for $\pred_t$. By Lemma~\ref{lem:updating_predictions}, all edges in $G''$ must have already arrived by $C_{u,v}(\pred_0)-1$ in $\pred_0$, so $G'$ has every edge $G''$ has.  This means that $u$ and $v$ are not in the same strongly connected component in $G''$, which implies that $C_{u,v}(\pred_t) \geq C_{u,v}(\pred_0) - 2\eta$.
\end{proof}

We are ready to prove Theorem~\ref{thm:runtime}.
\begin{proof}[Proof of Theorem~\ref{thm:runtime}]
Each time a node is relabeled in the Node Label Data Structure, the size of its SCC doubles;  thus the total cost of all merges is $O(n\log n)$.

The cost to build a subproblem is linear in the number of vertices of $\G_x$ and the number of edges $\tilde{E}_x$.
Let $\mathcal{S}(t)$ be the set of subproblems built at time $t$.  The total cost to build all subproblems is
\[
\sum_{t} \sum_{t' \in \mathcal{S}(t)} (\# \text{ edges in } t') = 
\sum_{e\in \sigma} (\# \text{ builds of subproblems containing } e).
\]
We bound this final sum. 
First, any $e= (u,v)$ where $C_{u,v}(\pred_t)$ is undefined can only be in subproblems with $r = m$ by Lemma~\ref{lem:combining_time_in_subproblem}; there are at most $\log m$ such subproblems,  and each is rebuilt $O(\eta)$ times by Lemma~\ref{lem:rebuild_edges}, so the total cost of these edges is $O(m\eta\log m)$.

Now suppose $C_{u,v}(\pred_t)$ is defined and $e = (u,v)$ is in $[\ell, r]$ built at time $t$. We show that in each level, over time, there are $O(\eta)$ such subproblems. Summing over all $\log m$ levels gives a cost of $O(\eta m\log m)$.

If $e$ is in $[\ell,r]$ at time $t$, by Lemma~\ref{lem:combining_time_in_subproblem} we have $C_{u,v}(\pred_t)\in [\ell,r]$, and Lemma~\ref{lem:combining_time_changing} implies that $C_{u,v}(\pred_0) \in [\ell-2\eta, r+2\eta]$.

Consider a depth $d$, such that $r - \ell \geq \eta$ for subproblems at depth $d$.  There are at most five subproblems at depth $d$ that have $C_{u,v}(\pred_0)\in [\ell - 2\eta, r + 2\eta]$.  Each is rebuilt $O(\eta)$ times by Lemma~\ref{lem:rebuild_edges}; therefore, there are $O(\eta)$ subproblem builds at depth $d$ that contain $e$.

Now we consider a depth $d$ such that all subproblems have $r - \ell < \eta$.  
If a subproblem $[\ell,r]$ in this level contains $e$, then $C_{u,v}(\pred_0)\leq r+2\eta < \ell+3\eta$ and $C_{u,v}(\pred_0) \geq \ell-2\eta > r-3\eta$, which imply $\ell>C_{u,v}(\pred_0)-3\eta$ and $r<C_{u,v}(\pred_0)+3\eta$, respectively.
Then if $e'$ causes the subproblem $[\ell,r]$ to be rebuilt, by Lemma~\ref{lem:rebuild_edges}, $e'$ was predicted to arrive at time steps $[\ell,r+\eta]$ in $\pred$. The length of this interval is $O(\eta)$, which means that $e$ has been involved in $O(\eta)$ rebuilds at depth $d$ over time.
\end{proof}

\section{Correctness} \label{sec:correctness}
In this section, we prove correctness of our algorithm.  

The following lemma characterizes exactly what edges are present in each subproblem for a fixed prediction.

\begin{lemma}
\label{lem:edges_in_subproblem}
At any time $t$, for any subproblem $[\ell, r]$ in the offline recursive tree built on $\pred_t$, an edge $e = (u,v)\in \pred_t$ is in $[\ell, r]$ if and only if:
\begin{itemize}
    \item $r = m$ and $C_{u,v}(\pred_t)$ is undefined; or
    \item $C_{u,v}(\pred_t) \in [\ell, r]$.
\end{itemize}
\end{lemma}

\begin{proof}
    Lemma~\ref{lem:combining_time_in_subproblem} immediately gives one direction of the proof.

    For the other direction, note that every edge $e$ in a subproblem is either a right or left edge.  Immediately, $e$ is an edge in exactly one subproblem at every level of the tree. 
    If $C_{u,v}(\pred_t)$ is defined, only one subproblem at every level satisfies $C_{u,v}(\pred_t)\in [\ell, r]$. If $C_{u,v}(\pred_t)$ is undefined, 
    again only one subproblem at every level satisfies $r = m$. Lemma~\ref{lem:combining_time_in_subproblem} implies that that subproblem must contain $e$.
\end{proof}

Now, we show that we rebuild the subproblems necessary to accurately update our components.

\begin{lemma}
\label{lem:correct_rebuild_necessary}
    Consider a subproblem $[\ell, r]$ on the path of subproblems maintained by the algorithm such that $[\ell, r]$ is not built at time $t$.   Then for all $u,v$ such that $C_{u,v}(\pred_{t-1})$ is defined, $C_{u,v}(\pred_{t-1}) \in [\ell, r]$ if and only if $C_{u,v}(\pred_t)\in [\ell, r]$.
\end{lemma}
\begin{proof}
    Recall that the edges in $\pred_{t-1}$ and $\pred_t$ are the same except for those in $\{t, \ldots, \hat{t}\}$.  Immediately, the lemma is correct for any $C_{u,v}(\pred_{t-1})\notin \{t, \ldots, \hat{t}\}$.

    Recall also that the algorithm rebuilds the first interval $[\ell_1, r_1]$ on the current path of subproblems containing $\hat{t}$ and all descendants.  We are left to show the lemma for the ancestors of $[\ell_1, r_1]$.  
 Consider building the path from the root to $[\ell_1, r_1]$.
     Any such ancestor $[\ell, r]$ with midpoint $x$ has either $\ell < t \leq \hat{t} < x$, or $x < t \leq \hat{t} < r$; therefore, $\G_x$ is the same under $\pred_{t-1}$ and $\pred_t$, so the left and right edges are the same under $\pred_{t-1}$ and $\pred_t$; we are done by Lemma~\ref{lem:edges_in_subproblem}.
\end{proof}

Finally, we can show correctness.

\begin{lemma}
    \label{lem:union-find-correct}
    At time $t$, for any vertices $u$ and $v$, $L[u] = L[v]$ in the node label data structure if and only if $u$ and $v$ are in the same SCC in $G_t$.
\end{lemma}
\begin{proof}
We prove this by induction on $t$.  At $t=0$, all nodes have distinct labels. As there are no edges, each node is in its own SCC.  

Consider some $t\neq m$.  Let $[\ell, r]$ be the subproblem with midpoint $t$ in the path of subproblems maintained by the algorithm.  Assume the lemma for time $t-1$: $L[u] = L[v]$ if and only if $C_{u,v}(\sigma) \leq t-1$.  Note that $\pred_{t-1}$ and $\sigma$ agree on the first $t-1$ edges, so we can equivalently write $C_{u,v}(\pred_{t-1}) \leq t-1$.  

Recall that our algorithm finds all SCCs of $\G_t$; for each SCC $C$, the algorithm merges all pairs of vertices so that they have the same label.  Our goal is to show that two vertices are merged in this way if and only if $C_{u,v}(\pred_t) = t$.

By Lemmas~\ref{lem:edges_in_subproblem} and~\ref{lem:correct_rebuild_necessary}, the edges of subproblem $[\ell, r]$ are exactly those with combining time in $[\ell, r]$ under $\pred_t$.  Thus, $\G_t$ consists of all edges $e = (u,v)$ such that: $C_{u,v}(\pred_t) \in [\ell, r]$, and $e$ arrives\footnote{Specifically, $e$ arrives in $\sigma$: the true sequence of arrivals.} on or before $t$. 

Consider two nodes $u$ and $v$ with $C_{u,v}(\pred_t) = t$; there must be a path $p$ from $u$ to $v$ in $G_t$.  Our goal is essentially to show that there is a path from $u$ to $v$ in $\G_t$.
Note that since $u$ and $v$ are in the same SCC at time $t$, any two nodes on the path $p$ must be in the same SCC at time $t$.

For each edge $(w, z)$ in $p$ we split into two cases.  If $C_{w, z}(\pred_t) < \ell$, then $w$ and $z$ were merged into a single vertex in an ancestor of $[\ell, r]$.  If $C_{w, z}(\pred_t) \geq \ell$, by the above we have that $w$ and $z$ are in the same SCC at time $t$, so $C_{w, z}(\pred_t) \leq t \leq r$.  Since $(w, z)$ in $p$ is an edge in $G_t$, it arrives by time $t$ in $\pred_t$, so $(w, z)\in \G_t$.  

We immediately obtain a path $p'$ in $\G_t$ by merging successive vertices $(w, z)$ in $p$ with $C_{w, z}(\pred_t) < \ell$; by the inductive hypothesis these merged vertices all have the same label.  Similarly, there is a path $p''$ in $\G_t$ obtained by merging successive vertices in the path from $v$ to $u$ in $G_t$.  Thus, all vertices in $p'$ and $p''$ will be merged in the Node Label Data Structure, and $L[u] = L[v]$ at time $t$.

Now, the other direction: we show that if $L[u] = L[v]$ at time $t$, then $u$ and $v$ are in the same SCC of $G_t$.   By the inductive hypothesis, we need only show this for pairs $u$ and $v$ that are merged in the Node Label Data Structure at time $t$.   

A pair $u$ and $v$ are merged if and only if there is a path from $u$ to $v$, and a path from $v$ to $u$ in $\G_t$.
Let $p$ be the path from $u$ to $v$ in $\G_t$.  Each edge in $p$ is an edge that arrives before $t$ in $\G_t$, so each edge in $p$ is also an edge in $G_t$.  Note that $\G_t$ is a condensed graph, so each node represents a set of nodes in $G_t$.
Consider two successive edges $e_1 = (u_1,v_1)$ and $e_2 = (u_2, v_2)$ in $p$ (here, $u_1, v_1, u_2, v_2\in V$---we are referring to $e_1$ and $e_2$ using their endpoints in $G_t$).  Then $v_1$ and $u_2$ were merged into a single vertex in $[\ell, r]$; this occurs if $v_1$ and $u_2$ were in the same SCC of some subproblem $[\ell', r']$ with midpoint $x'$, where $[\ell, r]$ is a right descendant of $[\ell', r']$.  Since $[\ell, r]$ is a right descendant of $[\ell', r']$, $x' < \ell$.  Therefore, $v_1$ and $u_2$ must be in the same SCC at time $x' < \ell$.  Thus, there is a path from $v_1$ to $u_2$ in $G_t$.  

Thus, each edge in path $p$ in $\G_t$ is an edge in $G_t$, and between the original endpoints of two successive edges in $p$ there is a path in $G_t$. 
Therefore there is a path from $u$ to $v$ in $G_t$. Using the same argument, there is a path from $v$ to $u$ in $G_t$ as well. 
Thus, $u$ and $v$  are in the same SCC of $G_t$.
\end{proof}
\section{Experimental Evaluation}\label{sec:experiments}
In this section, we describe our experimental results.\footnote{The Python implementations  of our experiments are available at \href{https://github.com/helia-niaparast/incremental-strongly-connected-components-with-predictions}{https://github.com/helia-niaparast/incremental-strongly-connected-components-with-predictions}.}

\subparagraph*{Incremental SCC$^+$.}
We compare against the state-of-the-art heuristic for the problem: $\text{IncSCC}^+$~\cite{fan2017incremental}.
In $\text{IncSCC}^+$, the algorithm maintains a topological ordering of the nodes at all times, that is, each node $v$ in the current (possibly contracted) graph is assigned a \emph{topological rank} $r(v)$ such that for every arc $(u,u')$ between distinct SCCs, we have $r(u) > r(u')$. When a new arc $(v,w)$ is added in the ``wrong'' direction, i.e., $r(v) < r(w)$, the algorithm identifies potentially affected nodes by running forward and backward depth-first searches from $w$ and $v$, respectively. It then applies Tarjan’s algorithm on the affected nodes to determine whether a new cycle has been formed.   We were not able to obtain the authors' original code, so we implemented our own version of their algorithm. 
In this implementation, when a cycle is formed, we merge the SCCs in the cycle into a new SCC, and run DFS on the whole graph to find the new topological ordering.

We also made an optimized version (IncSCC+ optimized) with two modifications.
First, we skip the final run of Tarjan's algorithm, as cycle formation can be detected immediately from the affected nodes.  Second, after a node is contracted, in order to find a new valid topological ordering, we only rearrange the affected nodes. 
These factors do not change the main ideas of the algorithm, but lead to a noticeable performance improvement in practice.  
In our tests, we compare to both $\text{IncSCC}^+$ implementations.

We also include the performance of the \textbf{offline incremental algorithm} from Section~\ref{sec:warmup}, which has access to the entire edge arrival sequence ahead of time.  

\subparagraph*{Datasets.}
Our datasets are from SNAP\footnote{\url{https://snap.stanford.edu/data/}}~\cite{snapnets}.  The
\texttt{sx-superuser-c2a} and
\texttt{sx-askubuntu} datasets are graphs based on two stackexchange forums; a link represents one user posting an answer or comment on another user's question or answer. These two are temporal, i.e., links are ordered based on the time the comments occurred.  
The other dataset used, \texttt{Slashdot0811}, is not temporal; edges arrive in random order.  
Each edge represents a user who tagged another as ``friend'' or ``foe.''
The number of vertices and number of unique edges for each of these datasets is given in Table~\ref{tab:results}.

\subparagraph*{Experimental Setup.}
Our implementations are in Python. Our experiments were run on an Ubuntu Linux machine with a 2.8GHz i7-1165G7 Intel CPU with 32GB RAM.

\subparagraph*{Predictions.} We parameterize how to generate predictions with a variable $S$, which is roughly the standard deviation of the error.
For each value of $S$, we construct a perturbed edge sequence from the original ordering as follows. We traverse the edges in their original order, and for each edge, we sample an offset from a normal distribution with mean $0$ and standard deviation $S$. The target index is set to the current index plus this offset. If both the current and target indices have not yet been modified, we swap the edges at these positions; otherwise, the edge remains fixed. 

After processing all edges, we measure the prediction quality using two metrics relative to the original sequence: the average prediction error ($\eta_{\text{avg}}$) and the maximum prediction error ($\eta_{\text{max}}$).
For each $S \in \{0, 10, 20, \ldots\}$, we generate 10 independent perturbed sequences using the procedure described above. Each sequence yields its own average and maximum error values. We then run the algorithm on each sequence and report the average runtime and error across the 10 runs. In the plots, $\eta_{\text{avg}}$ is the mean of the 10 average-error values and $\eta_{\text{max}}$ is the mean of the 10 maximum-error values.

\begin{figure*}[t]
    \centering
    \hfill
    \includegraphics[width=0.47\textwidth]{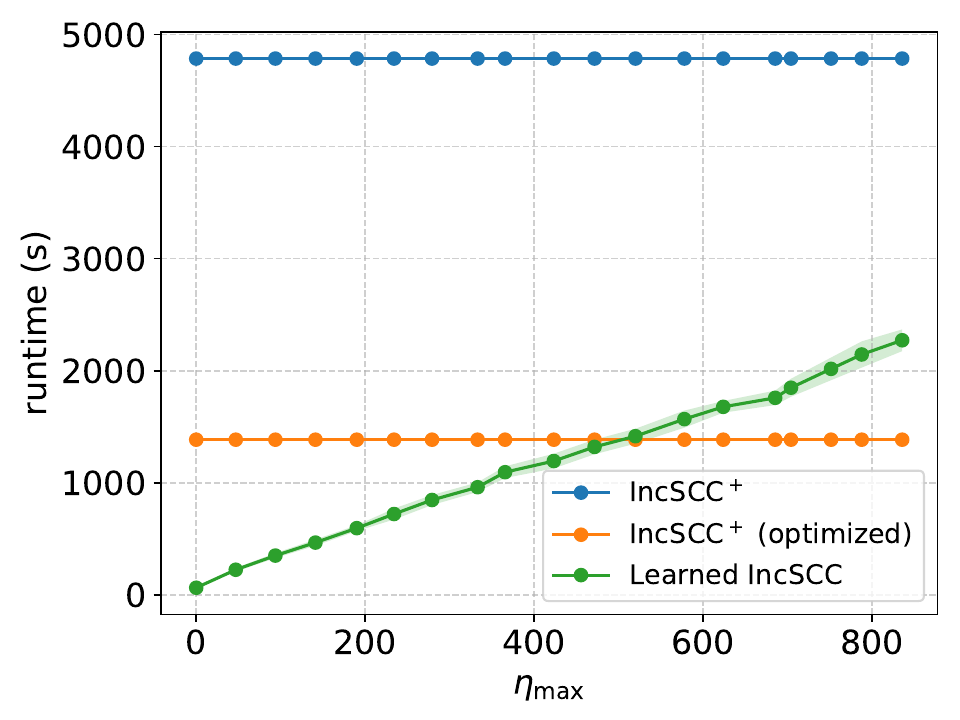}
    \hfill
    \includegraphics[width=0.47\textwidth]{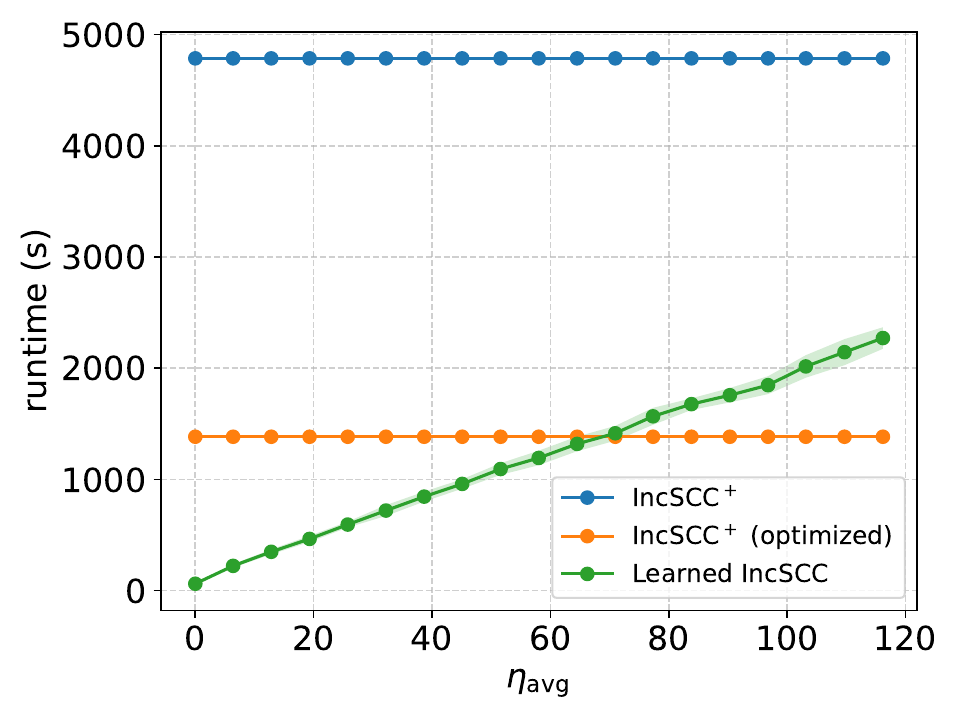}
    \hfill
    \vspace{-.1in}
    \caption{Runtime comparison on the \texttt{sx-askubuntu} dataset.
In the left and right plots, the x-axis denotes the maximum and average prediction error with respect to the actual edge sequence, respectively.
The green line shows the mean runtime, and the surrounding shaded region indicates the standard deviation.}
\label{fig:twoplots}
\end{figure*}

\begin{table*}[t]
\centering
\footnotesize
\setlength{\tabcolsep}{2pt}
\renewcommand{\arraystretch}{1.2}
\begin{tabular}{|c|c|c|c|c|c|c|c|c|}
\hline
\rowcolor{gray!20}
Dataset & $m$ & $n$ & Offline &\makecell{Learned \\ IncSCC \\ $\eta = 0$} & $\text{IncSCC}^+$ & \makecell{$\text{IncSCC}^+$ \\ (optimized)} & \makecell{{Crossover} \\ ${\eta}_{\text{avg}}$} & \makecell{{Crossover}\\${\eta}_{\text{max}}$}\\
\hline
\texttt{sx-superuser-c2a} & 289487 & 101052 & 60.70 & 34.21 & 932.10 & 331.33 & $\approx 32$ & $\approx 219$ \\
\hline
\texttt{sx-askubuntu} & 596933 & 159316 & 85.37 & 62.24 & 4786.33 & 1384.01 & $\approx 70$ & $\approx 517$\\
\hline
\texttt{Slashdot0811} & 905468 & 77360 & 71.33 & 47.92 & 2221.04 & 362.68 & $\approx 20$ & $\approx 150$\\
\hline
\end{tabular}
\vspace{.1in}
\caption{Comparison of the performance of Learned IncSCC against the baselines. The first two columns show the number of unique edges and nodes. The next four columns report the runtimes of the purely offline algorithm, Learned IncSCC with perfect predictions, and the two versions of $\text{IncSCC}^{+}$. Crossover $\eta_{\text{avg}}$ and crossover $\eta_{\text{max}}$ are the largest average and max errors for which Learned IncSCC outperforms $\text{IncSCC}^{+}$ (optimized).}
\label{tab:results}
\vspace{-.1in}
\end{table*}

\subparagraph*{Discussion.}
We include the plot for \texttt{sx-askubuntu} in Figure~\ref{fig:twoplots}, and include the results for the remaining datasets in Appendix~\ref{sec:appendix-plots}.  
With sufficiently accurate predictions, our algorithm gives significantly improved performance over both versions of $\text{IncSCC}^+$.  The experiments show our performance grows roughly linearly with the error; in fact it grows linearly with both the $\eta_{\text{avg}}$ and $\eta_{\text{max}}$ (our theoretical analysis bounds performance only in terms of $\eta_{\text{max}}$). 

Table~\ref{tab:results} summarizes further results, including the largest average and maximum $\eta$ for which our algorithm improves on optimized $\text{IncSCC}^+$.  

In Table~\ref{tab:results}, we also compare our algorithm with perfect predictions to the offline incremental algorithm from Section~\ref{sec:warmup}.  While both algorithms are given the entire edge sequence, our algorithm can adapt to errors on-the-fly,  while the offline algorithm is static.  Despite this overhead, our running time is not only comparable but in fact slightly better than the offline algorithm.  This demonstrates that our method is lightweight, while still allowing for potential error. 
\section{Conclusion}
In this work, we design a new learned algorithm to maintain the strongly connected components in an incremental graph.  Our theoretical bounds show that our algorithm is nearly optimal under reasonably good predictions and its performance degrades smoothly with respect to the prediction error.  

We test our algorithm against the best heuristic on real-world datasets. Our experiments show that (1) theory is predictive of practice, and (2) it is possible to achieve strong provable guarantees against bad predictions, while also being able to match (and even beat) the practical performance of state-of-the-art heuristics when predictions are good.

\bibliographystyle{plainurl}
\bibliography{scc}

@inproceedings{chen2024almost,
  title={Almost-linear time algorithms for incremental graphs: Cycle detection, sccs, st shortest path, and minimum-cost flow},
  author={Chen, Li and Kyng, Rasmus and Liu, Yang P and Meierhans, Simon and Probst Gutenberg, Maximilian},
  booktitle={Proceedings of the 56th Annual ACM Symposium on Theory of Computing},
  pages={1165--1173},
  year={2024}
}

@article{haeupler2012incremental,
  title={Incremental cycle detection, topological ordering, and strong component maintenance},
  author={Haeupler, Bernhard and Kavitha, Telikepalli and Mathew, Rogers and Sen, Siddhartha and Tarjan, Robert E},
  journal={ACM Transactions on Algorithms (TALG)},
  volume={8},
  number={1},
  pages={1--33},
  year={2012},
  publisher={ACM New York, NY, USA}
}

@inproceedings{bernstein2019decremental,
  title={Decremental strongly-connected components and single-source reachability in near-linear time},
  author={Bernstein, Aaron and Probst, Maximilian and Wulff-Nilsen, Christian},
  booktitle={Proceedings of the 51st Annual ACM SIGACT Symposium on theory of computing},
  pages={365--376},
  year={2019}
}

@article{bender2015new,
  title={A new approach to incremental cycle detection and related problems},
  author={Bender, Michael A and Fineman, Jeremy T and Gilbert, Seth and Tarjan, Robert E},
  journal={ACM Transactions on Algorithms (TALG)},
  volume={12},
  number={2},
  pages={1--22},
  year={2015},
  publisher={ACM New York, NY, USA}
}

@inproceedings{fan2017incremental,
  title={Incremental graph computations: Doable and undoable},
  author={Fan, Wenfei and Hu, Chunming and Tian, Chao},
  booktitle={Proceedings of the 2017 ACM International Conference on Management of Data},
  pages={155--169},
  year={2017}
}

@article{mitzenmacher2022algorithms,
  title={Algorithms with predictions},
  author={Mitzenmacher, Michael and Vassilvitskii, Sergei},
  journal={Communications of the ACM (CACM)},
  volume={65},
  number={7},
  pages={33--35},
  year={2022},
  publisher={ACM New York, NY, USA}
}

@article{zhang2017efficient,
  title={Efficient disk-based directed graph processing: A strongly connected component approach},
  author={Zhang, Yu and Liao, Xiaofei and Shi, Xiang and Jin, Hai and He, Bingsheng},
  journal={IEEE Transactions on Parallel and Distributed Systems},
  volume={29},
  number={4},
  pages={830--842},
  year={2017},
  publisher={IEEE}
}

@article{dhingra2016finding,
  title={Finding strongly connected components in a social network graph},
  author={Dhingra, Swati and Dodwad, Poorvi S and Madan, Meghna},
  journal={International Journal of Computer Applications},
  volume={136},
  number={7},
  pages={1--5},
  year={2016},
  publisher={Citeseer}
}

@book{hardekopf2009pointer,
  title={Pointer analysis: building a foundation for effective program analysis},
  author={Hardekopf, Benjamin Charles},
  year={2009},
  publisher={The University of Texas at Austin}
}

@article{tarjan1972depth,
  title={Depth-first search and linear graph algorithms},
  author={Tarjan, Robert},
  journal={SIAM Journal on Computing},
  volume={1},
  number={2},
  pages={146--160},
  year={1972},
  publisher={SIAM}
}

@inproceedings{bernstein2021incremental,
  title={Incremental scc maintenance in sparse graphs},
  author={Bernstein, Aaron and Dudeja, Aditi and Pettie, Seth},
  booktitle={29th Annual European Symposium on Algorithms (ESA 2021)},
  year={2021}
}

@inproceedings{roditty2013decremental,
  title={Decremental maintenance of strongly connected components},
  author={Roditty, Liam},
  booktitle={Proceedings of the Twenty-Fourth Annual ACM-SIAM Symposium on Discrete Algorithms},
  pages={1143--1150},
  year={2013},
  organization={SIAM}
}

@InProceedings{McCauleyMoNi25,
  author =  {McCauley, Samuel and Moseley, Benjamin and Niaparast, Aidin and Niaparast, Helia and Singh, Shikha},
  title =  {Incremental Approximate Single-Source Shortest Paths with Predictions},
  booktitle =  {52nd International Colloquium on Automata, Languages, and Programming (ICALP)},
  pages =  {117:1--117:20},
  series =  {Leibniz International Proceedings in Informatics (LIPIcs)},
  ISBN =  {978-3-95977-372-0},
  ISSN =  {1868-8969},
  year =  {2025},
  volume =  {334},
  editor =  {Censor-Hillel, Keren and Grandoni, Fabrizio and Ouaknine, Jo\"{e}l and Puppis, Gabriele},
  publisher =  {Schloss Dagstuhl -- Leibniz-Zentrum f{\"u}r Informatik},
  address =  {Dagstuhl, Germany},
  URL =    {https://drops.dagstuhl.de/entities/document/10.4230/LIPIcs.ICALP.2025.117},
  URN =    {urn:nbn:de:0030-drops-234946},
  doi =    {10.4230/LIPIcs.ICALP.2025.117},
}

@inproceedings{lin2022learning,
  title={Learning augmented binary search trees},
  author={Lin, Honghao and Luo, Tian and Woodruff, David},
  booktitle={International Conference on Machine Learning},
  pages={13431--13440},
  year={2022},
  organization={PMLR}
}

@inproceedings{
bai2023sorting,
title={Sorting with Predictions},
author={Xingjian Bai and Christian Coester},
booktitle={Thirty-seventh Conference on Neural Information Processing Systems},
year={2023},
url={https://openreview.net/forum?id=Qv7rWR9JWa}
}

@inproceedings{
zeynali2024robust,
title={Robust Learning-Augmented Dictionaries},
author={Ali Zeynali and Shahin Kamali and Mohammad Hajiesmaili},
booktitle={Forty-first International Conference on Machine Learning},
year={2024},
url={https://openreview.net/forum?id=XyhgssAo5b}
}

@article{baswana2019efficient,
  title={An efficient strongly connected components algorithm in the fault tolerant model},
  author={Baswana, Surender and Choudhary, Keerti and Roditty, Liam},
  journal={Algorithmica},
  volume={81},
  number={3},
  pages={967--985},
  year={2019},
  publisher={Springer}
}

@inproceedings{SakaueO22,
  author    = {Shinsaku Sakaue and Taihei Oki},
  
  title     = {Discrete-convex-analysis-based framework for warm-starting algorithms with predictions},
  booktitle = {35th Conference on Neural Information Processing Systems (NeurIPS)},
  year      = {2022},
}

@inproceedings{LattanziSV23,
  author       = {Silvio Lattanzi and
                  Ola Svensson and
                  Sergei Vassilvitskii},
  editor       = {Andreas Krause and
                  Emma Brunskill and
                  Kyunghyun Cho and
                  Barbara Engelhardt and
                  Sivan Sabato and
                  Jonathan Scarlett},
  title        = {Speeding Up Bellman Ford via Minimum Violation Permutations},
  booktitle    = {International Conference on Machine Learning, {ICML} 2023, 23-29 July
                  2023, Honolulu, Hawaii, {USA}},
  series       = {Proceedings of Machine Learning Research},
  volume       = {202},
  pages        = {18584--18598},
  publisher    = {{PMLR}},
  year         = {2023},
  url          = {https://proceedings.mlr.press/v202/lattanzi23a.html},
  bibsource    = {dblp computer science bibliography, https://dblp.org}
}

@article{POLAK2024106487,
title = {Learning-augmented maximum flow},
journal = {Information Processing Letters},
volume = {186},
pages = {106487},
year = {2024},
issn = {0020-0190},
doi = {https://doi.org/10.1016/j.ipl.2024.106487},
url = {https://www.sciencedirect.com/science/article/pii/S0020019024000176},
author = {Adam Polak and Maksym Zub},
}

@InProceedings{DaviesMVW,
  title = 	 {Predictive Flows for Faster Ford-Fulkerson},
  author =       {Davies, Sami and Moseley, Benjamin and Vassilvitskii, Sergei and Wang, Yuyan},
  booktitle = 	 {Proc. of the 40th International Conference on Machine Learning (ICML)},
  pages = 	 {7231--7248},
  year = 	 {2023},
  editor = 	 {Krause, Andreas and Brunskill, Emma and Cho, Kyunghyun and Engelhardt, Barbara and Sabato, Sivan and Scarlett, Jonathan},
  volume = 	 {202},
  series = 	 {Proceedings of Machine Learning Research},
  month = 	 {23--29 Jul},
  publisher =    {PMLR},
  url = 	 {https://proceedings.mlr.press/v202/davies23b.html}
}

@inproceedings{DinitzILMV21,
  author    = {Michael Dinitz and
               Sungjin Im and
               Thomas Lavastida and
               Benjamin Moseley and
               Sergei Vassilvitskii},
  editor    = {Marc'Aurelio Ranzato and
               Alina Beygelzimer and
               Yann N. Dauphin and
               Percy Liang and
               Jennifer Wortman Vaughan},
  title     = {Faster Matchings via Learned Duals},
  booktitle = {Proc. 34th Conference on Advances in Neural Information Processing Systems (NeurIPS)},
  pages     = {10393--10406},
  year      = {2021},
  bibsource = {dblp computer science bibliography, https://dblp.org}
}

@inproceedings{KraskaBCDP18,
  author       = {Tim Kraska and
                  Alex Beutel and
                  Ed H. Chi and
                  Jeffrey Dean and
                  Neoklis Polyzotis},
  editor       = {Gautam Das and
                  Christopher M. Jermaine and
                  Philip A. Bernstein},
  title        = {The Case for Learned Index Structures},
  booktitle    = {Proc.\ 44th Annual International Conference on Management of Data, (SIGMOD)},
  pages        = {489--504},
  publisher    = {{ACM}},
  year         = {2018},
  url          = {https://doi.org/10.1145/3183713.3196909},
  doi          = {10.1145/3183713.3196909},
  bibsource    = {dblp computer science bibliography, https://dblp.org}
}

@InProceedings{liu2023predicted,
  title = 	 {The Predicted-Updates Dynamic Model: Offline, Incremental, and Decremental to Fully Dynamic Transformations},
  author =       {Liu, Quanquan C. and Srinivas, Vaidehi},
  booktitle = 	 {Proceedings of Thirty Seventh Conference on Learning Theory},
  pages = 	 {3582--3641},
  year = 	 {2024},
  editor = 	 {Agrawal, Shipra and Roth, Aaron},
  volume = 	 {247},
  series = 	 {Proceedings of Machine Learning Research},
  month = 	 {30 Jun--03 Jul},
  publisher =    {PMLR},
  url = 	 {https://proceedings.mlr.press/v247/liu24c.html}}

@inproceedings{
benomar2024learningaugmented,
title={Learning-Augmented Priority Queues},
author={Ziyad Benomar and Christian Coester},
booktitle={The Thirty-eighth Annual Conference on Neural Information Processing Systems},
year={2024},
url={https://openreview.net/forum?id=1ATLLgvURu}
}

@inproceedings{BrandFNP24,
  author       = {Jan van den Brand and
                  Sebastian Forster and
                  Yasamin Nazari and
                  Adam Polak},
  editor       = {David P. Woodruff},
  title        = {On Dynamic Graph Algorithms with Predictions},
  booktitle    = {Proceedings of the 2024 {ACM-SIAM} Symposium on Discrete Algorithms,
                  {SODA} 2024, Alexandria, VA, USA, January 7-10, 2024},
  pages        = {3534--3557},
  publisher    = {{SIAM}},
  year         = {2024},
  url          = {https://doi.org/10.1137/1.9781611977912.126},
  doi          = {10.1137/1.9781611977912.126},
  bibsource    = {dblp computer science bibliography, https://dblp.org}
}

@inproceedings{HenzingerSSY24,
  author       = {Monika Henzinger and
                  Barna Saha and
                  Martin P. Seybold and
                  Christopher Ye},
  editor       = {Venkatesan Guruswami},
  title        = {On the Complexity of Algorithms with Predictions for Dynamic Graph
                  Problems},
  booktitle    = {15th Innovations in Theoretical Computer Science Conference, {ITCS}
                  2024, January 30 to February 2, 2024, Berkeley, CA, {USA}},
  series       = {LIPIcs},
  volume       = {287},
  pages        = {62:1--62:25},
  publisher    = {Schloss Dagstuhl - Leibniz-Zentrum f{\"{u}}r Informatik},
  year         = {2024},
  url          = {https://doi.org/10.4230/LIPIcs.ITCS.2024.62},
  doi          = {10.4230/LIPICS.ITCS.2024.62},
  bibsource    = {dblp computer science bibliography, https://dblp.org}
}

@inproceedings{McCauleyMNS23,
  author       = {Samuel McCauley and Benjamin Moseley and Aidin Niaparast and Shikha Singh},
  title        = {Online List Labeling with Predictions},
  booktitle    = {Proc. 36th Conference on Neural Information Processing Systems (NeurIPS)},
  year         = {2023},
 editor = {A. Oh and T. Naumann and A. Globerson and K. Saenko and M. Hardt and S. Levine},
 pages = {60278--60290},
 publisher = {Curran Associates, Inc.},
 volume = {36}
}

@misc{sccblog,
  title = {My own algorithm — offline incremental strongly connected components in O(m log(m))},
  howpublished = {\url{https://codeforces.com/blog/entry/91608}},
  author = {Mateusz Radecki},
year = {2021},
  note = {Accessed: 2025-09-30}
}

@misc{snapnets,
    author= {Jure Leskovec and Andrej Krevl},
    title= {{SNAP Datasets}: {Stanford} Large Network Dataset Collection},
    howpublished = {\texttt{http://snap.stanford.edu/data}},
    month        = jun,
    year         = 2014
}

@inproceedings{
niaparast2025faster,
title={Faster Global Minimum Cut with Predictions},
author={Helia Niaparast and Benjamin Moseley and Karan Singh},
booktitle={Forty-second International Conference on Machine Learning},
year={2025},
url={https://openreview.net/forum?id=FeZoO7mfG7}
}

@inproceedings{
mccauley2024incremental,
title={Incremental Topological Ordering and Cycle Detection with Predictions},
author={Samuel McCauley and Benjamin Moseley and Aidin Niaparast and Shikha Singh},
booktitle={Forty-first International Conference on Machine Learning},
year={2024},
url={https://openreview.net/forum?id=wea7nsJdMc}
}

@article{dinitz2024binary,
  title={Binary search with distributional predictions},
  author={Dinitz, Michael and Im, Sungjin and Lavastida, Thomas and Moseley, Ben and Niaparast, Aidin and Vassilvitskii, Sergei},
  journal={Advances in Neural Information Processing Systems},
  volume={37},
  pages={90456--90472},
  year={2024}
}

@article{mccauley2026stable,
  title={Stable Matching with Predictions: Robustness and Efficiency under Pruned Preferences},
  author={McCauley, Samuel and Moseley, Benjamin and Niaparast, Helia and Singh, Shikha},
  journal={arXiv preprint arXiv:2602.02254},
  year={2026}
}

\newpage
\appendix
\section{Additional Plots}\label{sec:appendix-plots}

In this section, we include the plots for \texttt{sx-superuser-c2a} and \texttt{Slashdot0811} datasets. The setup is similar to Figure~\ref{fig:twoplots}. 

\begin{figure*}[h]
    \centering
    \hfill
    \includegraphics[width=0.47\textwidth]{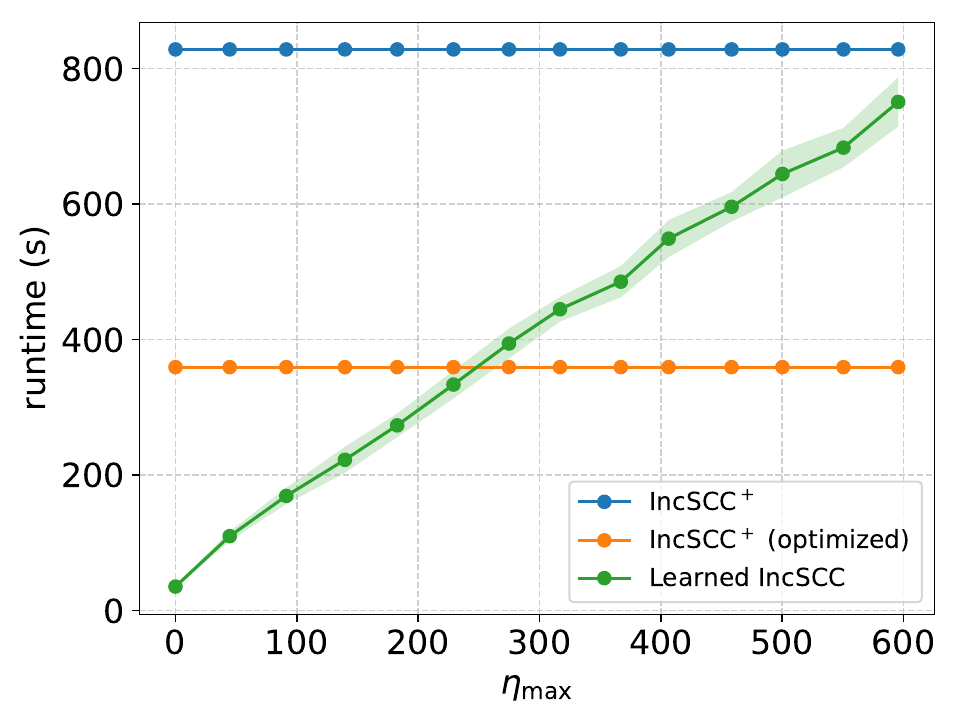}
    \hfill
    \includegraphics[width=0.47\textwidth]{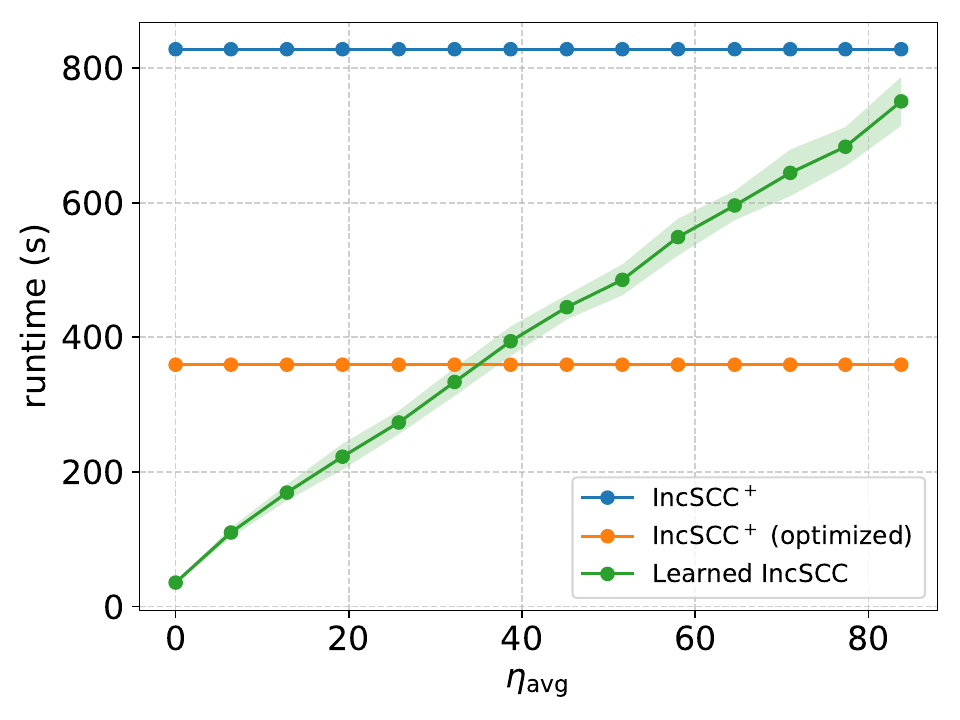}
    \hfill
    \vspace{-.1in}
    \caption{Runtime comparison on the \texttt{sx-superuser-c2a} dataset.}
\label{fig:superuser}
\end{figure*}
\begin{figure*}[h]
    \centering
    \hfill
    \includegraphics[width=0.47\textwidth]{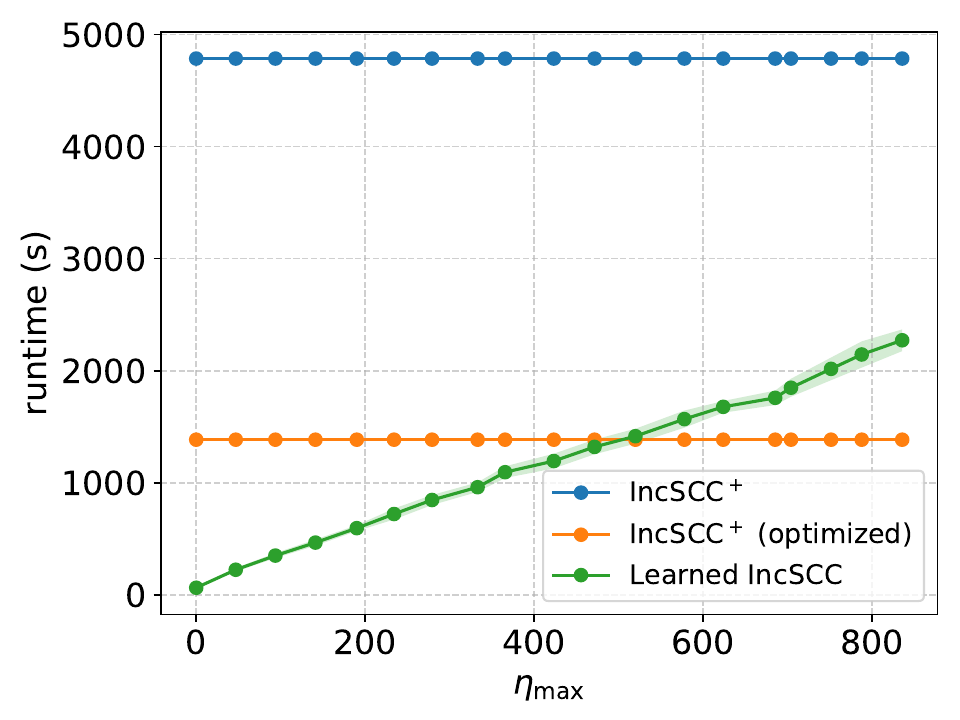}
    \hfill
    \includegraphics[width=0.47\textwidth]{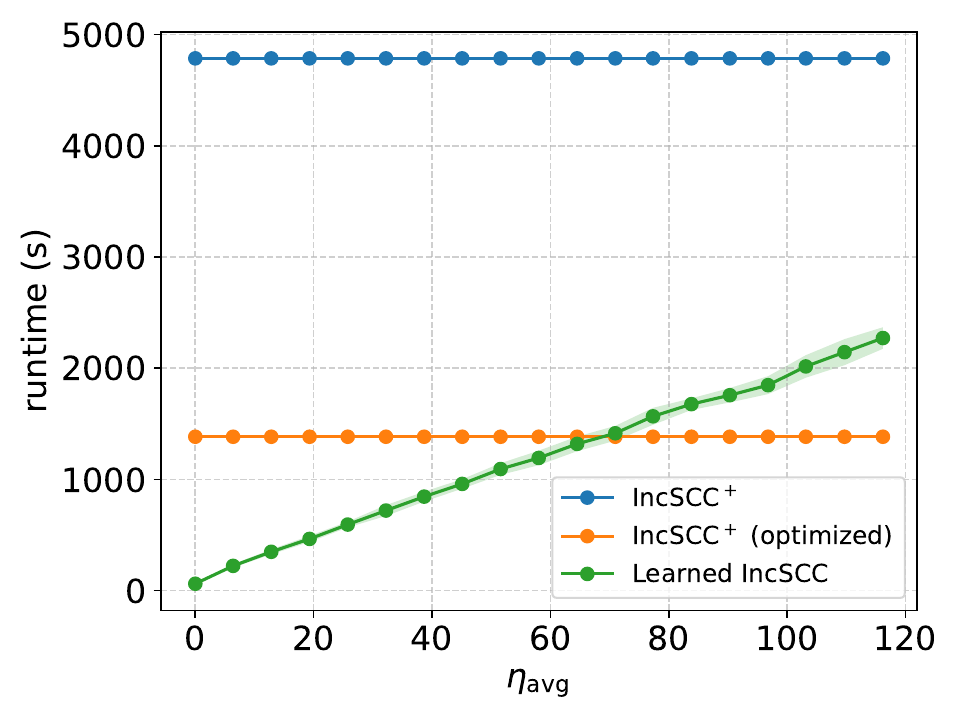}
    \hfill
    \vspace{-.1in}
    \caption{Runtime comparison on the \texttt{Slashdot0811} dataset.}
\label{fig:slashdot}
\end{figure*}

\end{document}